\documentclass[12pt]{article}

\usepackage{amsmath}
\usepackage{amssymb}
\usepackage{amsfonts}
\usepackage{latexsym}
\usepackage{color}

\catcode `\@=11 \@addtoreset{equation}{section}

\catcode `\@=12

\newcommand{\be}{\begin{equation}}
\newcommand{\en}{\end{equation}}
\newcommand{\bea}{\begin{eqnarray}}
\newcommand{\ena}{\end{eqnarray}}
\newcommand{\beano}{\begin{eqnarray*}}
\newcommand{\enano}{\end{eqnarray*}}
\newcommand{\bee}{\begin{enumerate}}
\newcommand{\ene}{\end{enumerate}}

\newcommand{\N}{\mathfrak N}
\newcommand{\M}{\mathfrak M}

\newcommand{\mc}{\mathcal}

\newcommand{\D}{{\mc D}}

\newcommand{\F}{{\cal F}}

\newcommand{\I}{{\mathbb{I}}}

\newcommand{\1}{1 \!\! 1}

\newcommand{\Hil}{\mc H}

\newtheorem{thm}{Theorem}
\newtheorem{cor}[thm]{Corollary}

\newtheorem{defn}[thm]{Definition}

\newenvironment{proof}{\noindent {\bf Proof --}}{\hfill$\square$ \vspace{3mm}\endtrivlist}

\catcode `\@=11 \@addtoreset{equation}{section}
\catcode `\@=12

\begin{document}

\thispagestyle{empty}

\vspace*{2cm}

\begin{center}
{\Large \bf  Non linear pseudo-bosons}\\[10mm]

{\large F. Bagarello}\\
  Dipartimento di Metodi e Modelli Matematici,
Facolt\`a di Ingegneria,\\ Universit\`a di Palermo, I-90128  Palermo, Italy\\
e-mail: bagarell@unipa.it\\ Home page:
www.unipa.it/$^\sim$bagarell\\

\vspace{3mm}

\end{center}

\vspace*{2cm}

\begin{abstract}
\noindent In a series of recent papers the author has introduced the notion of (regular) pseudo-bosons showing, in particular, that two number-like operators, whose spectra are ${\Bbb N}_0:={\Bbb N}\cup\{0\}$, can be naturally introduced. Here we extend this construction to operators with rather more general spectra. Of course, this generalization can be applied to many more physical systems. We  discuss several examples of our framework.

\end{abstract}

\vspace{2cm}


\vfill


\newpage

\section{Introduction}

In a series of recent papers \cite{bagpb1,bagpb2,bagpb3,bagcal,bagpbJPA}, we have investigated some mathematical aspects of the
so-called {\em pseudo-bosons} (PB),  originally introduced by Trifonov
in \cite{tri}. They arise from the canonical commutation relation
$[a,a^\dagger]=\1$ upon replacing $a^\dagger$ by another (unbounded)
operator $b$ not (in general) related to $a$: $[a,b]=\1$. We have
shown that, under suitable assumptions, $N=ba$ and $N^\dagger=a^\dagger b^\dagger$ can be both
diagonalized, and that their spectra coincide with the set of
natural numbers (including 0), ${\Bbb N}_0$. However the sets of
related eigenvectors are not orthonormal (o.n) bases but,
nevertheless, they are automatically {\em biorthogonal}. In most of the
examples considered so far, they are bases of the Hilbert space of the system,
$\Hil$, and, in some cases, they turn out to be {\em Riesz bases}.

In \cite{bagpb4} and \cite{abg} some physical examples arising from concrete examples in quantum mechanics have been discussed. In particular, the difference between {\em regular pseudo-bosons} (RPB) and PB has been introduced: RPB, see Section II, arise when the  eigenvectors of $N$ and $N^\dagger$ are mapped the ones into the others by a bounded operator with bounded inverse. If this operator is unbounded, then our PB are {\em not regular}. The strong limit of our construction is that it may be used successfully only for those operators, self-adjoint or not, for which the n-th eigenvalue $\lambda_n$ is linear in $n$: $\lambda_n=\omega\, n+k$, where $\omega$ and $k$ are real constants. As we have shown in \cite{bagpb4} and \cite{abg}, as well as in \cite{bagrev}, this is exactly what happens in several interesting physical systems. However, the spectra of most quantum mechanical hamiltonians is not linear in $n$ so that it is natural to wonder which part of the structure of RPB, if any, can be extended to include these more general operators. Indeed this is possible, and this extension is the main result of this paper.

The paper is organized as follows:  in the next section we
introduce and discuss some features of $d$-dimensional PB.  In Sections III we show how non-linear coherent states suggest a possible extension of RPB to operators with non-linear spectra. This will originate new excitations which we will call {\em non-linear pseudo bosons} (NLPB) and {\em non-linear regular pseudo bosons} (NLRPB). We will see that a similar structure as that for PB is recovered. Section IV is devoted to examples, while our conclusions are given in Section V.

\section{$d$-dimensional PB and RPB}

In this section, for completeness sake, we review our construction for PB and RPB, which can be found in more details in \cite{bagpb1,bagpbJPA}.

Let $\Hil$ be a given Hilbert space with scalar product
$\left<.,.\right>$ and related norm $\|.\|$. We introduce $d$
pairs of operators, $a_j$ and $b_j$, $j=1,2,\ldots,d$, acting on $\Hil$ and
satisfying the following commutation rules \be [a_j,b_j]=\1,
\label{21} \en where $j=1,2,\ldots,d$, while all the other commutators are trivial. Of course, they collapse to the CCR's for $d$
independent modes if $b_j=a^\dagger_j$, $j=1,2,\ldots,d$. Since $a_j$ and $b_j$ are unbounded operators,  they cannot be
defined on all of $\Hil$. Following \cite{bagpb1}, and writing
$D^\infty(X):=\cap_{p\geq0}D(X^p)$ (the common  domain of all the powers of the
operator $X$), we consider the
following:

\vspace{2mm}

{\bf Assumption 1.--} there exists a non-zero
$\varphi_{\bf 0}\in\Hil$ such that $a_j\varphi_{\bf 0}=0$, $j=1,2,\ldots,d$,
and $\varphi_{\bf 0}\in D^\infty(b_1)\cap D^\infty(b_2)\cap\cdots\cap D^\infty(b_d)$.

{\bf Assumption 2.--} there exists a non-zero $\Psi_{\bf 0}\in\Hil$
such that $b_j^\dagger\Psi_{\bf 0}=0$, $j=1,2,\ldots,d$, and $\Psi_{\bf 0}\in
D^\infty(a_1^\dagger)\cap D^\infty(a_2^\dagger)\cap\cdots\cap D^\infty(a_d^\dagger)$.

\vspace{2mm}

Under these assumptions we can introduce the following vectors in
$\Hil$:

\be
\left\{
\begin{array}{ll}
\varphi_{\bf n}:=\varphi_{n_1,n_2,\ldots,n_d}=\frac{1}{\sqrt{n_1!n_2!\cdots n_d!}}\,b_1^{n_1}\,b_2^{n_2}\cdots b_d^{n_d}\,\varphi_{\bf 0}\\
\Psi_{\bf n}:=\Psi_{n_1,n_2,\ldots,n_d}=\frac{1}{\sqrt{n_1!n_2!\cdots n_d!}}\,{a_1^\dagger}^{n_1}\,{a_2^\dagger}^{n_2}\cdots {a_d^\dagger}^{n_d}\,\Psi_{\bf 0},
\end{array}
\right.\label{22}\en $n_j=0, 1, 2,\ldots$ for all $j=1,2,\ldots,d$. Let us now define the unbounded
operators $N_j:=b_ja_j$ and $\N_j:=N_j^\dagger=a_j^\dagger
b_j^\dagger$, $j=1,2,\ldots,d$.  It is possible to check that
$\varphi_{\bf n}$ belongs to the domain of $N_j$, $D(N_j)$, and
$\Psi_{\bf n}\in D(\N_j)$, for all possible $\bf n$. Moreover,
\be N_j\varphi_{\bf n}=n_j\varphi_{\bf n},  \quad \N_j\Psi_{\bf n}=n_j\Psi_{\bf n}.
\label{23}\en

Under the above assumptions, and if we chose the normalization of
$\Psi_{\bf 0}$ and $\varphi_{\bf 0}$ in such a way that
$\left<\Psi_{\bf 0},\varphi_{\bf 0}\right>=1$, we find that \be
\left<\Psi_{\bf n},\varphi_{\bf m}\right>=\delta_{\bf n,m}=\prod_{j=1}^d \delta_{n_j,m_j}. \label{27}\en Then the sets
$\F_\Psi=\{\Psi_{\bf n}\}$ and
$\F_\varphi=\{\varphi_{\bf n}\}$ are {\em biorthogonal} and,
because of this, the vectors of each set are linearly independent.
If we now call $\D_\varphi$ and $\D_\Psi$ respectively the linear
span of  $\F_\varphi$ and $\F_\Psi$, and $\Hil_\varphi$ and
$\Hil_\Psi$ their closures,  we make the

\vspace{2mm}

{\bf Assumption 3.--} The above Hilbert spaces all coincide:
$\Hil_\varphi=\Hil_\Psi=\Hil$.

\vspace{2mm}

This means, in particular,
that both $\F_\varphi$ and $\F_\Psi$ are bases of $\Hil$. Let us
now introduce the operators $S_\varphi$ and $S_\Psi$ via their
action respectively on  $\F_\Psi$ and $\F_\varphi$: \be
S_\varphi\Psi_{\bf n}=\varphi_{\bf n},\qquad
S_\Psi\varphi_{\bf n}=\Psi_{\bf n}, \label{213}\en for all $\bf n$, which also imply that
$\Psi_{\bf n}=(S_\Psi\,S_\varphi)\Psi_{\bf n}$ and
$\varphi_{\bf n}=(S_\varphi \,S_\Psi)\varphi_{\bf n}$, for all
$\bf n$. Hence \be S_\Psi\,S_\varphi=S_\varphi\,S_\Psi=\1 \quad
\Rightarrow \quad S_\Psi=S_\varphi^{-1}. \label{214}\en In other
words, both $S_\Psi$ and $S_\varphi$ are invertible and one is the
inverse of the other. Furthermore,  they are
both positive, well defined and symmetric, \cite{bagpb1}. Moreover, it is possible to write these operators using the
bra-ket notation as \be S_\varphi=\sum_{\bf n}\,
|\varphi_{\bf n}><\varphi_{\bf n}|,\qquad S_\Psi=\sum_{\bf n}
\,|\Psi_{\bf n}><\Psi_{\bf n}|. \label{212}\en
 These expressions are
only formal, at this stage, since the series may not converge in
the uniform topology and the operators $S_\varphi$ and $S_\Psi$ could be unbounded.
Indeed we know,  \cite{you}, that two biorthogonal bases are related by a bounded operator, with bounded inverse, if and only if they are Riesz bases\footnote{Recall that a set of vectors $\phi_1, \phi_2 , \phi_3 , \; \ldots \; ,$ is a Riesz basis of a Hilbert space $\mathcal H$, if there exists a bounded operator $V$, with bounded inverse, on $\mathcal H$, and an orthonormal basis of $\Hil$,  $\varphi_1, \varphi_2 , \varphi_3 , \; \ldots \; ,$ such that $\phi_j=V\varphi_j$, for all $j=1, 2, 3,\ldots$}. This is why we have also considerered

\vspace{2mm}

{\bf Assumption 4.--} $\F_\varphi$ and $\F_\Psi$ are Riesz bases of $\Hil$.

\vspace{3mm}
\noindent
This assumption  implies  that $S_\varphi$ and $S_\Psi$ are bounded operators, whose domains can be taken to be all of $\Hil$. Physical motivations have suggested to call {\em pseudo-bosons} (PB) those satisfying  the first three assumptions, and {\em regular pseudo-bosons} (RPB) those satisfying also Assumption 4, \cite{bagrev}.

It is easy to
check that \be S_\Psi\,N_j=\N_jS_\Psi \quad \mbox{ and }\quad
N_j\,S_\varphi=S_\varphi\,\N_j, \label{219}\en $j=1,2,\ldots,d$. This is
related to the fact that the spectra of, say, $N_1$ and $\N_1$
coincide and that their eigenvectors are related by the operators
$S_\varphi$ and $S_\Psi$, in agreement with the literature on
intertwining operators.

\section{Non-linear regular pseudo-bosons}

In this section we will show how essentially the same general framework of RPB can be recovered in a different situation, that is when the operators $N$ and $\N$ have a non-linear spectra. The idea behind this generalization is the same which produces, starting from {\em standard} coherent states, {\em non-linear coherent states}, see \cite{book2} and references therein.

\subsection{Non-linear coherent states}

In the old literature, see \cite{ks} for instance, a standard
coherent state is a vector arising from the action of the
unitary operator $U(z)=e^{z\,a^\dagger-\overline{z}\,a}$,
$z\in\mathbb{C}$ and $[a,a^\dagger]=\1$, on the vacuum of $a$,
$\Phi_0$, $a\Phi_0=0$: $|z>=U(z)\Phi_0$. These normalized vectors
can be written in other equivalent ways, introducing the o.n.
basis $\{\Phi_n,\,n\in\mathbb{N}_0\}$ where
$\Phi_n=\frac{(a^\dagger)^n}{\sqrt{n!}}\,\Phi_0$, as follows: \be
|z>=U(z)\Phi_0=e^{-|z|^2/2}e^{z\,a^\dagger}\,\Phi_0=e^{-|z|^2/2}\,\sum_{k=0}^\infty
\,\frac{z^n}{\sqrt{n!}}\,\Phi_n.\label{51}\en It is well known that $|z>$ are  eigenstate
of $a$: $a|z>=z|z>$, and that they satisfy a {\em resolution of the
identity}: $\frac{1}{\pi}\,\int\,d^2z\,|z><z|=\I$. They also
saturate the Heisenberg uncertainty principle: let
$q=\frac{a+a^\dagger}{\sqrt{2}}$,
$p=\frac{a-a^\dagger}{i\,\sqrt{2}}$, $(\Delta X)^2=<z,X^2z>-<z,X\,z>^2$
for $X=q,p$. Then $\Delta q\,\Delta p=\frac{1}{2}$.

These properties are recovered using a different definition for
the coherent states, see \cite{book2,ali} and references therein,
which generalizes the one above.
Starting from a sequence $\{\epsilon_l,\,l\in\mathbb{N}_0\}$, of non negative
numbers, $\epsilon_l\geq 0$ for all $l\in\mathbb{N}_0$, it is possible to define
some vectors, parametrized by a complex $z$,  as follows: \be
\Xi(z):=N(|z|^2)^{-1/2}\,\sum_{k=0}^\infty
\,\frac{z^n}{\sqrt{\epsilon_n!}}\,\Phi_n,\label{52}\en where
$N(|z|^2)=\sum_{k=0}^\infty \,\frac{|z|^{2n}}{\epsilon_n!}$. With this choice we have
$<\Xi,\Xi>=1$ for all $|z|\leq\rho$, $\rho$ being the radius of
convergence of the series for $N$, and where $\epsilon_0!=1$ and
$\epsilon_n!=\epsilon_1\,\epsilon_2\ldots \epsilon_n$. In particular, if $A$ is an operator satisfying $A\Phi_n=\sqrt{\epsilon_n}\,\Phi_{n-1}$,
 we deduce that $A\,\Xi(z)=z\,\Xi(z)$, so that these
generalized coherent states  are again eigenstates of the
(generalized) annihilation operator $A$. Moreover, in \cite{ali} it
is also shown that the existence of a resolution of the identity,
that is the existence of a measure $d\nu(z,\overline{z})$ such
that
$\int_{C_\rho}\,N(|z|^2)\,|\Xi(z)><\Xi(z)|\,d\nu(z,\overline{z})=\1$,
 is related to the existence of a solution of the
following {\em moment problem}: we put $z=r\,e^{i\theta}$,
$d\nu(z,\overline{z})=d\theta\,d\lambda(r)$,
$C_\rho=\{z=r\,e^{i\theta}, \,\theta\in[0,2\pi[, r\in[0,\rho[\}$,
then we want $d\lambda(r)$ to be such that \be
\int_0^{\rho}\,d\lambda(r)\,r^{2k}=\frac{x_k!}{2\pi}, \quad
\forall k\in{\Bbb N}_0. \label{53}\en It is known that this problem has
not always solution, but when it does, then a resolution of the
identity can be established.

Finally, if we introduce two self-adjoint operators (the
generalized position and momentum operators)
$Q=\frac{A+A^\dagger}{\sqrt{2}}$,
$P=\frac{A-A^\dagger}{i\,\sqrt{2}}$, then $\Xi(z)$ saturates
the Heisenberg uncertainty principle, which now can be written as
\be \Delta Q\,\Delta
P=\frac{1}{2}\left|<AA^\dagger>-|z|^2\right|.\label{54}\en Notice
that, if $x_k=k$, we recover the coherent states in (\ref{51}).

\subsection{Non linear RPB}

Using a similar idea as that producing non-linear coherent states we will now show how the framework in Section II can be extended. The first difficulty is that, since in general $\epsilon_n\neq n$, the commutation rule $[a,b]=\1$ is not expected to hold anymore. Here $\epsilon_n$ can be seen as the real n-th eigenvalue of a non necessarily self-adjoint operator, see below.  Since $\epsilon_n$ does not depend linearly on $n$, in general, it is  convenient to  consider the following definition, which is slightly different from the one given for RPB. We begin with assuming that the sequence $\{\epsilon_n\}$ is strictly increasing and that $\epsilon_0=0$: $0=\epsilon_0<\epsilon_1<\cdots<\epsilon_n<\cdots$. Then, given two operators $a$ and $b$ on the Hilbert space $\Hil$,

\begin{defn}
We will say that the triple $(a,b,\{\epsilon_n\})$ is a family of non-linear regular pseudo-bosons (NLRPB) if the following properties hold:
\begin{itemize}

\item {\bf p1.} a non zero vector $\Phi_0$ exists in $\Hil$ such that $a\,\Phi_0=0$ and $\Phi_0\in D^\infty(b)$.

\item {\bf { p2}.} a non zero vector $\eta_0$ exists in $\Hil$ such that $b^\dagger\,\eta_0=0$ and $\eta_0\in D^\infty(a^\dagger)$.

\item {\bf { p3}.} Calling
\be
\Phi_n:=\frac{1}{\sqrt{\epsilon_n!}}\,b^n\,\Phi_0,\qquad \eta_n:=\frac{1}{\sqrt{\epsilon_n!}}\,{a^\dagger}^n\,\eta_0,
\label{55}
\en
we have, for all $n\geq0$,
\be
a\,\Phi_n=\sqrt{\epsilon_n}\,\Phi_{n-1},\qquad b^\dagger\eta_n=\sqrt{\epsilon_n}\,\eta_{n-1}.
\label{56}\en
\item {\bf { p4}.} The sets $\F_\Phi=\{\Phi_n,\,n\geq0\}$ and $\F_\eta=\{\eta_n,\,n\geq0\}$ are bases of $\Hil$.

\item {\bf { p5}.} $\F_\Phi$ and $\F_\eta$ are Riesz bases of $\Hil$.

\end{itemize}

\end{defn}

{\bf Remarks:--} (1) Notice that the definitions in (\ref{55}) are well posed in the sense that, because of {\bf p1} and {\bf p2}, the vectors $\Phi_n$ and $\eta_n$ are well defined for all $n\geq0$: indeed we have $\Phi_0\in D^\infty(b)$ and $\eta_0\in D^\infty(a^\dagger)$.

(2) If {\bf {\bf p5}} is not satisfied, but the others do, then  we call our particles non-linear pseudo-bosons (NLPB).

(3) But for {\bf p3}, the other conditions above coincide exactly with those of RPB. In fact, we can show that {\bf p3} replaces (and extends) the commutation rule $[a,b]=\1$, which is recovered if $\epsilon_n=n$. Indeed in this case from (\ref{55}) we deduce that $b\,\Phi_n=\sqrt{n+1}\,\Phi_{n+1}$ and $a^\dagger\eta_n=\sqrt{n+1}\,\eta_{n+1}$, while hypothesis (\ref{56}) implies that  $a\,\Phi_n=\sqrt{n+1}\,\Phi_{n+1}$ and $b^\dagger\eta_n=\sqrt{n+1}\,\eta_{n+1}$. Hence $[a,b]\,\Phi_n=\Phi_n$, for all $n\geq0$, so that, being $\F_\Phi$ a basis because of assumption {\bf p4}, $[a,b]=\1$ follows.

\vspace*{3mm}

This last remark shows that NLPB are indeed extensions of PB. Hence  similar results are expected. Indeed, let us introduce the following (not self-adjoint) operators:
\be
M=ba,\qquad \M=M^\dagger=a^\dagger b^\dagger.
\label{57}\en
Then we can check that $\Phi_n\in D(M)\cap D(b)$, $\eta_n\in D(\M)\cap D(a^\dagger)$, and, more than this, that
\be
b\,\Phi_n=\sqrt{\epsilon_{n+1} }\,\Phi_{n+1},\qquad a^\dagger\eta_n=\sqrt{\epsilon_{n+1} }\,\eta_{n+1},
\label{58}\en
which is a consequence of definitions (\ref{55}), see also Remark (3) above, as well as
\be
M\Phi_n=\epsilon_n\Phi_n,\qquad \M\eta_n=\epsilon_n\eta_n,
\label{59}\en
These eigenvalue equations have a very important consequence, close to that deduced for PB: the vectors in $\F_\Phi$ and $\F_\eta$ are mutually orthogonal. More explicitly,
\be
\left<\Phi_n,\eta_m\right>=\delta_{n,m},\left<\Phi_0,\eta_0\right>
\label{510}\en
The proof of this equation does not differ significantly from that for RPB, and will not be given here.

Let us now consider the following condition:

\vspace{2mm}

{\bf {\bf p3}$'$.} The vectors $\Phi_n$ and $\eta_n$ defined in (\ref{55}) satisfy (\ref{510}).

\vspace{2mm}

Then it is possible to check that conditions  \{{\bf p1}, {\bf p2}, {\bf p3}, {\bf p4}\} are equivalent to  \{{\bf p1}, {\bf p2}, {\bf p3}$'$, {\bf p4}\}. We see that {\bf p5} plays no role here. In one direction the implication is clear: each NLPB satisfies {\bf p3}$'$. Viceversa, suppose \{{\bf p1}, {\bf p2}, {\bf p3}$'$, {\bf p4}\} are satisfied. Then, using (\ref{58}) and (\ref{510}), we have $\left<a\Phi_n,\eta_k\right>=\left<\Phi_n,a^\dagger\eta_k\right>=\sqrt{\epsilon_{k+1}}
\left<\Phi_n,\eta_{k+1}\right>=\sqrt{\epsilon_{k+1}}\delta_{n,k+1}\left<\Phi_0,\eta_0\right>$. On the other hand, since $\F_\Phi$ is a basis for $\Hil$, we can expand $a\Phi_n$ as follows: $a\Phi_n=\sum_{l}\,d_l^{(n)}\,\Phi_l$, where the coefficients depend in general on both $n$ and $l$. In particular, since $a\Phi_0=0$, $d_l^{(0)}=0$ for all $l$. Then, using again (\ref{510}), we get $\left<a\Phi_n,\eta_k\right>=\sum_{l}\,\overline{d_l^{(n)}}\,\left<\Phi_l,\eta_{k}\right>=
\overline{d_k^{(n)}}\,\left<\Phi_0,\eta_{0}\right>$. Hence $d_k^{(n)}=\sqrt{\epsilon_{k+1}}\,\delta_{n,k+1}$ and, consequently, $a\Phi_n=\sqrt{\epsilon_n}\,\Phi_{n-1}$, which is the first equation in (\ref{56}). The second equation can be deduced in a similar way: then, {\bf p3} is recovered. In the following, therefore, we can use {\bf p3} or {\bf p3}$'$ depending of which is more convenient for us.

\vspace{2mm}

Carrying on our analysis on the consequences of the definition on NLRPB, and in particular of ${\bf p4}$, we rewrite this assumption in bra-ket formalism as
\be
\sum_n|\Phi_n><\eta_n|=\sum_n|\eta_n><\Phi_n|=\1,
\label{512}\en
while {\bf p5} implies that the operators $S_\Phi:=\sum_n|\Phi_n><\Phi_n|$ and $S_\eta:=\sum_n|\eta_n><\eta_n|$ are positive, bounded, and invertible. Notice that these definitions are the right ones if $\left<\Phi_0,\eta_0\right>=1$, which we will assume from now on, otherwise an extra normalization factor should be considered. Moreover, as in \cite{bagpb1}, it is possible to show that $S_\Phi=S_\eta^{-1}$. The new fact is that the operators $a$ and $b$ do not, in general, satisfy any {\em simple} commutation rule. Indeed, we can check that, for all $n\geq0$,
\be
[a,b]\Phi_n=\left(\epsilon_{n+1}-\epsilon_n\right)\Phi_n,
\label{513}\en
which returns (\ref{21}) for $d=1$ if $\epsilon_n=n+\alpha$, for any real $\alpha$ but not in general. Moreover, we can also deduce that $[b^\dagger,a^\dagger]\eta_n=\left(\epsilon_{n+1}-\epsilon_n\right)\eta_n$.

In \cite{bagpbJPA} we have shown that any RPB is, in a certain sense, related to ordinary bosons. We extend here this result, showing that to each NLRPB can be associated a family of non-linear bosons and a bounded operator with bounded inverse. More in details we have

\begin{thm}
Let $(a,b,\{\epsilon_n\})$ be a family of NLRPB. Then there exist a positive operator $T\in B(\Hil)$, with bounded inverse,  an operator $c$ on $\Hil$, and an o.n. basis $\F_{\hat\Phi}=\{\hat\Phi_n\}$ of $\Hil$, with $\hat\Phi_n\in D(c)\cap D(c^\dagger)$, $\forall n\geq0$, such that
\be
[c,c^\dagger]\hat\Phi_n=\left(\epsilon_{n+1}-\epsilon_n\right)\hat\Phi_n,
\label{514}\en
and
\be
a=TcT^{-1},\qquad b=Tc^\dagger T^{-1}.
\label{515}\en
Viceversa, let us consider an operator $c$, an o.n. basis of $\Hil$, $\F_{\hat\Phi}=\{\hat\Phi_n\}$, and a sequence $\{\epsilon_n\}$ such that $0=\epsilon_0<\epsilon_1<\cdots$. Let us assume that, for all $n$, $c\,\hat\Phi_n=\sqrt{\epsilon_{n}}\,\hat\Phi_{n-1}$. Then any positive bounded operator $T$, with bounded inverse, produce two operators $a$ and $b$ as is (\ref{515}) such that $(a,b,\{\epsilon_n\})$ is a family of NLRPB.
\label{th1}\end{thm}

The proof of this theorem is not significantly different from that for RPB given in \cite{bagpbJPA}, and therefore will not be given here. Once again, the operator $T$ cited above is nothing but the operator $S_\Phi^{1/2}$, where $S_\Phi$, already introduced above, is the frame operator of the Riesz basis $\F_\Phi$. If {\bf p5} is not assumed (i.e. if we consider NLPB which are not regular), $\F_\Phi$ is no longer a Riesz basis and, as a consequence,  $S_\Phi$ is an unbounded operator with unbounded inverse, in general. In \cite{bagpbJPA} we have considered this problem as well, and we have shown that a similar result can be stated also in these milder conditions. We claim that a similar extension also holds in the present context, but we will postpone the details of this analysis to a future paper. As for the nature of the operators $a$, $b$ and $c$ involved here we can prove easily the following

\begin{cor}
Let $(a,b,\{\epsilon_n\})$ be a family of NLRPB with $\sup_n\epsilon_n=\infty$. Then the operators $a$ and $b$ are unbounded.
\end{cor}
\begin{proof}
Our previous theorem shows that, starting from $(a,b,\{\epsilon_n\})$, we can construct a third operator $c$ and an o.n. basis $\F_{\hat\Phi}=\{\hat\Phi_n\}$ such that $c\,\hat\Phi_n=\sqrt{\epsilon_n }\Phi_{n-1}$. This implies, in particular, that $\|c\hat\Phi_n\|^2=\epsilon_n$, so that
$$
\|c\|=\sup_{\|f\|=1}\|cf\|\geq \sup_n\|c\hat\Phi_n\|=\sup_n\sqrt{\epsilon_n }=\infty.
$$
Then, since $a$ and $c$ are related as in (\ref{515}), it is clear that $a$ is unbounded as well. The same conclusion applies for $b$.
\end{proof}

An important feature of PB is the existence of an operator which intertwines between $N$ and $\N$, see (\ref{219}). The same intertwining relations can be deduced also in this framework: let us call $M_0=c^\dagger c$. Then, if we work under the assumption of Theorem \ref{th1}, we deduce easily that
\be
T\M=M_0T,\qquad MT=TM_0,
\label{516}\en
which also imply, recalling that $T^{-1}$ exists in $\Hil$, also that $MT^2=T^2\M$.

\section{Examples}

This section is devoted to some examples of our framework. We begin with something directly related to non-linear coherent states. Then we discuss an example in a finite-dimensional Hilbert space. After that we construct a general (and quite abstract) strategy which produces NLRPB. We end this section showing how quons fit into our settings.

\subsection{A first class of examples}

Let $a$ and $b$ be two operators satisfying Assumptions 1-4 of Section II, and let us define two new operators $A:=a$ and $B:=f(N)b$, where $N=ba$ and $f(x)$ is a {\em sufficiently regular} function of $x$ like, for instance, an odd polynomial in $x$. We have also to require that $f(0)=0$ and that $f(x)\geq0$ for $x\geq0$. We will always assume in the following that our formulas make sense, and we will add some comments when needed. It is clear that $[A,B]\neq\1$ in general. It is interesting to consider the simplest case: $f(N)=N$. In this case the commutators of $A$, $B$ and $N$ close: $[A,B]=2N+\1$, $[N,A]=-A$, $[N,B]=B$, but this is not true for other choices of $f$.

Let us see now if {\bf p1}-{\bf p5} are satisfied or not. First of all it is possible to check that, if we put  $\Phi_0=\varphi_0$, where $\varphi_0$ is the one in Assumption 1, then $A\Phi_0=0$ and $\Phi_0\in D^\infty(B)$. The first equality is obvious. The second follows from the following identity:
\be
B^n\Phi_0=[f(n)]!\sqrt{n!}\,\varphi_n,
\label{ex1}\en
for all $n\geq0$. Here $\varphi_n=\frac{b^n}{\sqrt{n!}}\,\varphi_0$, see Section II, and $[f(n)]!=f(1)f(2)\cdots f(n)$. The proof of this result goes via induction on $n$. Since, by construction, each $\varphi_n$ is in $\Hil$, then (\ref{ex1}) implies that the vectors $\Phi_n=\frac{1}{\sqrt{\epsilon_n!}}\,B^n\Phi_0=\frac{[f(n)]!\sqrt{n!}}{\sqrt{\epsilon_n!}}\,\varphi_n$  are also well defined in $\Hil$ for all $n\geq0$ and for all possible choices of $\epsilon_n$.

As for {\bf p2}, if we introduce a vector $\eta_0=\Psi_0$, we see that $B^\dagger\eta_0=b^\dagger f(\N)\Psi_0=0$ since $\N\Psi_0=0$ and $f(0)=0$. Moreover, $\eta_0\in D^\infty(A^\dagger)$. In fact we have, for all $n\geq0$, ${A^\dagger}^n\eta_0=\sqrt{n!}\,\Psi_n$. Then we put $\eta_n=\frac{1}{\sqrt{\epsilon_n!}}\,{A^\dagger}^n\eta_0=\sqrt{\frac{n!}{\epsilon_n!}}\,\Psi_n$, which are  well defined vectors in $\Hil$ for all $n\geq0$ and for all possible choices of $\epsilon_n$.

General reasons now imply that the families $\F_\Phi=\{\Phi_n,\,n\geq0\}$ and   $\F_\eta=\{\eta_n,\,n\geq0\}$ should be biorthogonal:
$\left<\Phi_n,\eta_m\right>=\delta_{n,m}\left<\Phi_0,\eta_0\right>$. This produces some constraint on the possible choices of $\epsilon_n$. It is in fact easy to check that $\left<\Phi_n,\eta_m\right>=\frac{n![f(n)]!}{\epsilon_n!}\delta_{n,m}\left<\Phi_0,\eta_0\right>$, so that the previous biorthogonality is recovered choosing $\epsilon_n=nf(n)$, but not in general. With this choice we deduce that
\be
\Phi_n=\sqrt{[f(n)]!}\,\varphi_n,\qquad \eta_n=\frac{1}{[f(n)]!}\,\Psi_n,
\label{ex2}\en
for all $n\geq0$. A simple computation now shows that
$$
A\Phi_n=a\sqrt{[f(n)]!}\,\varphi_n=\sqrt{[f(n)]!}\,\sqrt{n}\varphi_{n-1}=\sqrt{[f(n)]!}\,\sqrt{n}\frac{\Phi_{n-1}}{\sqrt{[f(n-1)]!}}=
\sqrt{\epsilon_n}\Phi_{n-1}
$$
and, analogously, that
$$
B^\dagger\eta_n=
\sqrt{\epsilon_n}\eta_{n-1}
$$
for all $n\geq0$. Hence {\bf p3} is also satisfied. Condition {\bf p4} is clearly also satisfied since the two sets $\F_\Phi$ and   $\F_\eta$ are proportional to the bases $\F_\varphi$ and   $\F_\Psi$, respectively. But we do not expect in general these sets to be also Riesz bases, due to the the presence of $\sqrt{[f(n)]!}$ in the normalization of $\Phi_n$. If we could know that $\|\varphi_n\|$ doesn't go to zero for $n\rightarrow\infty$, then we could conclude that $\F_\Phi$ is not a Riesz basis. However, this is not the case, in general. So we can only claim that {\em there are indications suggesting that {\bf p5} is not satisfied}.

Let us now introduce the operators $M=BA=f(N)N$ and $\M=M^\dagger=\N f(\N)$. Then $M\Phi_n=\epsilon_n\Phi_n$ and $\M\eta_n=\epsilon_n\eta_n$, for all $n\geq0$. Finally, we have (formally, in general),
$$
\left\{
\begin{array}{ll}
S_\Phi=\sum_{n=0}^\infty|\Phi_n><\Phi_n|=\sum_{n=0}^\infty[f(n)]!\,|\varphi_n><\varphi_n|,\\
S_\eta=\sum_{n=0}^\infty|\eta_n><\eta_n|=\sum_{n=0}^\infty\frac{1}{[f(n)]!}\,|\Psi_n><\Psi_n|,
\end{array}
\right.
$$
as well as
$$
\sum_{n=0}^\infty|\Phi_n><\eta_n|=\sum_{n=0}^\infty|\eta_n><\Phi_n|=\1,
$$
and $S_\Phi=S_\eta^{-1}$. Also, using $M=\sum_{n=0}^\infty\epsilon_n\,|\Phi_n><\eta_n|$ and $\M=\sum_{n=0}^\infty\epsilon_n\,|\eta_n><\Phi_n|$, we can easily check that $M S_\Phi=S_\Phi\M$.

\subsection{Extensions of this example}

In \cite{tava} and in many other papers, non linear coherent states have been obtained replacing $a$ and $a^\dagger$, $[a,a^\dagger]=\1$, with $A=ah(\hat n)$ and $A^\dagger=h(\hat n) a^\dagger$, where $\hat n=a^\dagger a$ and $h(x)$ is a sufficiently regular function. In analogy with that, we now define the  self-adjoint operator $H_1=A^\dagger A=h(\hat n)^2\hat n$.  By introducing two new operators $B=h(\hat n)^2 a^\dagger$ and $A=a$ the hamiltonian can be rewritten as $H_1=BA$, and this appears to be a special case of the example in IV.1, with $f(x)$ replaced by $h^2(x)$ and $b=a^\dagger$. Hence, if $h(0)=0$, we are sure that our previous results hold and conditions {\bf p1}-{\bf p4} of Section III are satisfied.

The same conclusions also hold if we define $A$ and $B$ as $B=h(\hat n)^2 b$ and $A=a$, where  $[a,b]=\1$ and where Assumptions 1-4 of Section II are satisfied for $a$ and $b$: we are back to NLPB, not necessarily regular.

The last easy extension we want to discuss here arise from a given similarity transformation: let $A$ and $B$ be non-linear pseudo-bosonic operators satisfying {\bf p1}-{\bf p5}, and let $S$ be a bounded, self-adjoint operator. Then, if we define $\hat A=e^SAe^{-S}$ and  $\hat B=e^SBe^{-S}$, the same properties are satisfied. It is enough to introduce $\hat\Phi_0=e^S\Phi_0$ and $\hat\eta_0=e^{-S}\eta_0$, which are well defined due to our assumptions on $S$. Moreover, it is easy to check that, not surprisingly, $\hat\Phi_n=e^S\Phi_n$ and $\hat\eta_n=e^{-S}\eta_n$, for all $n\geq0$. Needless to say, also the other properties are satisfied. It should be mentioned that if we don't require $S$ to be bounded, then the situation becomes less regular since, for instance, the set of vectors $\{e^S\Phi_n\}$ is not a Riesz basis of $\Hil$.

\subsection{An example in a finite dimensional Hilbert space}

In many papers, \cite{findim}, examples of pseudo-symmetric quantum mechanics are considered in finite dimensional Hilbert spaces. The reason is clear: every operator and every vector can be explicitly constructed and the assumptions are easily checked. For this same reason we construct now such an example.

Let $\Hil={\Bbb C}^2$ be our Hilbert space and let us consider the following matrices on $\Hil$
$$
A=\left(
   \begin{array}{cc}
     -1 & \beta \\
     -\frac{1}{\beta} & 1 \\
   \end{array}
 \right),\qquad B=\left(
   \begin{array}{cc}
     -1 & \delta \\
     -\frac{1}{\delta} & 1 \\
   \end{array}
 \right),\qquad
$$
where $\beta, \delta$ are real quantities  and $\beta\neq\delta$ to prevent the two matrices to commute. It is easy to check that two non zero vectors $\Phi_0$ and $\eta_0$ do exist such that $A\Phi_0=B^\dagger\eta_0=0$. These vectors are $\Phi_0=y\left(
           \begin{array}{c}
             \beta \\
             1 \\
           \end{array}
         \right)
$
and $\eta_0=w\left(
           \begin{array}{c}
             1 \\
             -\delta \\
           \end{array}
         \right)
$, where $y$ and $w$ are  normalization constants which we take real and satisfying $yw(\beta-\delta)=1$. It is clear that we can act with any power of $B$ on $\Phi_0$ and with any power of $A^\dagger$ on $\eta_0$. Hence {\bf p1} and {\bf p2} are satisfied. For all possible positive choices of $\epsilon_1$ (recall that $\epsilon_0$ must be zero), we can define
$$
\Phi_1=\frac{1}{\sqrt{\epsilon_1}}\,B\,\Phi_0=\frac{y}{\sqrt{\epsilon_1}}\left(
           \begin{array}{c}
             \delta-\beta \\
             -\frac{\beta}{\delta}+1 \\
           \end{array}
         \right), \qquad \eta_1=\frac{1}{\sqrt{\epsilon_1}}\,A^\dagger\,\eta_0=\frac{w}{\sqrt{\epsilon_1}}\left(
           \begin{array}{c}
             \frac{\delta}{\beta}-1 \\
              \beta-\delta\\
           \end{array}
         \right).
$$
It is easy to check that both $\F_\Phi=\{\Phi_0,\Phi_1\}$ and $\F_\eta=\{\eta_0,\eta_1\}$ are linearly independent in $\Hil$, if $\beta\neq\delta$. Hence they are bases. More than this: they are Riesz bases, since they can be obtained from the canonical o.n. basis of $\Hil$ via the action of two bounded operators with bounded inverses. It is also easy to check that $\left<\Phi_0,\eta_1\right>=\left<\Phi_1,\eta_0\right>=0$, so that the sets $\F_\Phi$ and $\F_\eta$ are biorthogonal.

Formula (\ref{56}) now fixes the form of $\epsilon_1$. Indeed it is possible to check that, while $A\Phi_0=B^\dagger\eta_0=0$ by construction, $A\Phi_1=\sqrt{\epsilon_1}\Phi_0$ and $B^\dagger\eta_1=\sqrt{\epsilon_1}\,\eta_0$ only if $\epsilon_1=-\frac{1}{\beta\delta}(\beta-\delta)^2$. With this choice, calling
$$
M=BA=\left(
       \begin{array}{cc}
         1-\frac{\delta}{\beta} & \delta-\beta \\
         \frac{1}{\delta}-\frac{1}{\beta} & -\frac{\beta}{\delta}+1 \\
       \end{array}
     \right)\qquad \M=A^\dagger B^\dagger=\left(
       \begin{array}{cc}
         1-\frac{\delta}{\beta} & \frac{1}{\delta}-\frac{1}{\beta} \\
         \delta-\beta & -\frac{\beta}{\delta}+1 \\
       \end{array}
     \right),
$$
it is easy to check that $M\Phi_k=\epsilon_k\Phi_k$ and $\M\eta_k=\epsilon_k\eta_k$, $k=0,1$. It is also easy to compute $[A,B]$, which is different from zero if $\delta\neq\beta$ and is never equal to the identity operator. Also, the resolution of the identity $\sum_{k=0}^1|\Phi_k><\eta_k|=\1$ holds true and we further find
$$
S_\Phi=\sum_{k=0}^1|\Phi_k><\Phi_k|=y^2\left(
       \begin{array}{cc}
         \beta(\beta-\delta) & 0 \\
         0 & 1-\frac{\beta}{\delta} \\
       \end{array}
     \right),$$
and
     $$S_\eta=\sum_{k=0}^1|\eta_k><\eta_k|=w^2\left(
       \begin{array}{cc}
         1-\frac{\delta}{\beta} & 0 \\
         0 & \delta(\delta-\beta) \\
       \end{array}
     \right).
$$
A direct computation finally shows that $S_\Phi=S_\eta^{-1}$ and that, as in the previous example, $MS_\Phi =S_\Phi \M$.

\subsection{A general class of examples}

Let $H$ be a self-adjoint hamiltonian with eigenvectors $\Psi_n$ and corresponding eigenvalues $\epsilon_n$. We assume that $\F_\Psi=\{\Psi_n,\,n\geq0\}$ is an o.n. basis of $\Hil$ and that $0=\epsilon_0<\epsilon_1<\epsilon_2<\ldots$. We can introduce two operators, $a$ and $a^\dagger$, via their actions on $\F_\Psi$: $a\Psi_n=\sqrt{\epsilon_n}\,\Psi_{n-1}$ and $a^\dagger\Psi_n=\sqrt{\epsilon_{n+1}}\,\Psi_{n+1}$, $n\geq0$. Notice that this second equation follows from the first one and from the orthonormality of $\F_\Psi$. Then $H=a^\dagger a$. In bra-ket terms these operators can be written as
$$
a=\sum_{n=1}^\infty \sqrt{\epsilon_n}\,|\Psi_{n-1}><\Psi_n|,\quad a^\dagger=\sum_{n=0}^\infty \sqrt{\epsilon_{n+1}}\,|\Psi_{n+1}><\Psi_n|,\quad H=\sum_{n=0}^\infty \epsilon_n\,|\Psi_{n}><\Psi_n|.
$$
As in Section IV.2 we can now construct NLPB via similarity operators. Hence, let  $S$ be a bounded, self-adjoint operator with bounded inverse $S^{-1}$. Then, calling $A=SaS^{-1}$ and $B=Sa^\dagger S^{-1}$, the triple $(A,B,\{\epsilon_n\})$ produce NLRPB. In this case it is enough to take $\Phi_0=S\Psi_0$ and $\eta_0=S^{-1}\Psi_0$. Properties {\bf p1}-{\bf p5} are trivially satisfied. A particularly simple choice of $S$ is $S=\sum_{n=0}^\infty s_n\,|\Psi_{n}><\Psi_n|$, where $0<\alpha\leq s_n\leq\beta<\infty$ for all $n$. Any such sequence produce a different $S$ and, as a consequence, different NLRPB.

\subsection{Quons}

Let $c$ and $c^\dagger$ satisfy the q-mutator $[c,c^\dagger]_q:=c\,c^\dagger-q\,c^\dagger\, c=\1$. Here $q\in[-1,1]$. These operators where introduced in \cite{moh} and analyzed along the years by several author. Recently they have also been used in connection with the theory of intertwining operators in \cite{bagquons}. Let $\varphi_0$ be the vacuum of $c$: $c\,\varphi_0=0$.
In \cite{moh} it is proved that the eigenstates of $N_0=c^\dagger c$ are analogous to the bosonic ones, but for the normalization. More in details, putting
\be\varphi_n=\frac{1}{\beta_0\cdots\beta_{n-1}}\,{c^\dagger}^n\,\varphi_0\qquad n\geq 0,\label{q1}\en
we have $N_0\varphi_n=\epsilon_n\varphi_n$, with $\epsilon_0=0$, $\epsilon_1=1$ and $\epsilon_n=1+q+\cdots+q^{n-1}$ for $n\geq 1$. Also, the normalization is found to be $\beta_n^2=1+q+\cdots+q^n$, for all $n\geq0$. Hence $\epsilon_n=\beta_{n-1}^2$ for all $n\geq1$.  The set of the $\varphi_n$'s  spans the Hilbert space $\Hil$ and they are mutually orhonormal: $<\varphi_{n},\varphi_{k}>=\delta_{n,k}$. Moreover, we also have
\be
c\,\varphi_{n}=\beta_{n-1}\varphi_{n-1}=\sqrt{\epsilon_n}\,\varphi_{n-1},
\label{q2}\en
which also implies that $c^\dagger\varphi_{n}=\sqrt{\epsilon_{n+1}}\,\varphi_{n+1}$, for all $n\geq 0$. We now introduce a non necessarily self-adjoint operator $S$, which here we take bounded for simplicity. Then, via similarity, we define $A=e^S\,c\,e^{-S}$, $B=e^S\,c^\dagger\,e^{-S}$, $\Phi_0:=e^S\varphi_0$ and $\eta_0=e^{-S^\dagger}\varphi_0$. Hence, conditions {\bf p1}-{\bf p5} are satisfied.

Let us now take, as an example, $S\equiv N_0$. Then $S$ is self-adjoint and bounded. Indeed, assuming for simplicity that $0<q<1$, we find that $\|N_0\|\leq\frac{1}{1-q}$. The operators $A$ and $B$ can be computed explicitly and we find that
$$
A=e^{N_0}c\,e^{-N_0}=e^{N_0(1-q)-\1}\,c,\qquad B=e^{N_0}c^\dagger e^{-N_0}=c^\dagger\,e^{\1+N_0(q-1)}.
$$
Moreover, since $\varphi_n$ are eigenstates of $N_0$, we also deduce that $\Phi_n=e^{\epsilon_n}\varphi_n$ and $\eta_n=e^{-\epsilon_n}\varphi_n$, for all $n\geq0$. It is clear that the related sets $\F_\Phi$ and $\F_\eta$ are biorthonormal and, due to the fact that for all $n\geq0$ we have $1\leq\epsilon_n<\frac{1}{1-q}$, $\F_\Phi$ and $\F_\eta$ are Riesz bases of $\Hil$.

\section{Conclusions}

We have shown that the notions of pseudo-bosons and regular pseudo-bosons can be generalized in order to include in our new settings operators with a discrete spectrum which is not linear in the { occupation number}. Surprisingly the main features of our previous construction do not change significantly: again biorthonormal bases of $\Hil$, which sometimes are Riesz bases, are recovered. These bases are eigenstates of two non self-adjoint operators related by an intertwining operator. We have also shown that this generalization is of the same kind which produces, starting from coherent states, the so-called non linear coherent states. Finally, we have discussed how and why quons fit very naturally within this settings.

\section*{Acknowledgements}
   The author acknowledge M.I.U.R. for financial support.


\begin{thebibliography}{99}

\bibitem{bagpb1} F. Bagarello, {\em Pseudo-bosons, Riesz bases and coherent states}, J. Math. Phys., {\bf 50}, DOI:10.1063/1.3300804, 023531 (2010) (10pg)

\bibitem{bagpb2} F. Bagarello {\em Construction of pseudo-bosons systems},  J. Math. Phys., {\bf 51}, doi:10.1063/1.3300804, 023531 (2010) (10pg)

\bibitem{bagpb3} F. Bagarello {\em Mathematical aspects of intertwining operators: the role of Riesz bases},  J. Phys. A, doi:10.1088/1751-8113/43/17/175203, {\bf 43},  175203 (2010) (12pp)

\bibitem{bagcal}  F. Bagarello, F. Calabrese {\em Pseudo-bosons arising from Riesz bases}, Bollettino del Dipartimento di Metodi e Modelli Matematici, {\bf 2}, 15-26, (2010)

\bibitem{bagpbJPA} F. Bagarello, {\em (Regular) pseudo-bosons versus bosons},  J. Phys. A, {\bf 44}, 015205 (2011)



\bibitem{tri} D.A. Trifonov, {Pseudo-boson coherent and Fock states}, quant-ph/0902.3744

\bibitem{bagpb4} F. Bagarello, {\em Examples of Pseudo-bosons in quantum mechanics},  Phys. Lett. A, {\bf 374}, 3823-3827 (2010)


\bibitem{abg} S.T. Ali, F. Bagarello, J.-P. Gazeau, {\em Modified Landau levels, damped harmonic oscillator and two-dimensional pseudo-bosons},  J. Math. Phys., {\bf 51}, 123502 (2010)

\bibitem{bagrev} F. Bagarello, {\em Pseudo-bosons, so far}, Rep. Math. Phys., submitted



\bibitem{you} Young R., {\em An introduction to nonharmonic Fourier series}, Academic Pree, New York, (1980)



\bibitem{book2}  J-P.  Gazeau, {\em Coherent states in quantum physics}, Wiley-VCH, Berlin 2009


 \bibitem{ks} J.R. Klauder, B.S. Skagerstam, {\em Coherent states- Applications to physics
 and mathematical physics}, World Scientific, Singapore (1985)

\bibitem{ali} S.T. Ali, M. Engli\v s and J.-P. Gazeau, {\em Vector coherent states from
       Plancherel's theorem, Clifford algebras and matrix domains\/}, J. Phys. {\bf A37},
       (2004), 6067-6089.

\bibitem{findim} A. Mostafazadeh, {\em Pseudo-symmetric quantum mechanics and isospectral pseudo-hermitian hamiltonians}, Nucl. Phys. B {\bf 640}, 419-434 (2002), K. Jones-Smith, H. Mathur, {\em A New Class of non-Hermitian Quantum Hamiltonians with PT symmetry}, hep-th:0908.4255, H. B. Geyer, W. D. Heiss, F. G. Scholtz, {\em Non-Hermitian Hamiltonians, Metric, Other Observables and Physical Implications}, quant-ph:0710.5593, A. Das, L. Greenwood, {\em An alternative construction of the positive inner product for pseudo-Hermitian Hamiltonians: Example}, Phys. Lett. B {\bf 678}, 504 (2009).

\bibitem{tava} M.K. Tavassoly, {\em New nonlinear coherent states associated with inverse bosonic and f-deformed ladder operators}, J. Phys. A, {\bf 41} (2008) 285305



\bibitem{moh} R.N. Mohapatra, {\em Infinite statistics and a possible small
violation of the Pauli principle}, Phys. Lett. B, {\bf 242}, 407-411, (1990); D.I. Fivel, {\em Interpolation between Fermi and Bose
statistics using generalized commutators}, Phys. Rev. Lett., {\bf 65},
3361-3364, (1990); Erratum, Phys. Rev. Lett., {\bf 69},
2020, (1992); O.W. Greenberg, {\em Particles with small violations of Fermi
or Bose statistics}, Phys. Rev. D, {\bf 43}, 4111-4120, (1991)

\bibitem{bagquons} F. Bagarello {\em Quons, coherent states and intertwining operators}, Phys. Lett. A, {\bf 373}, 2637-2642 (2009)




\end{thebibliography}
\end{document}